\documentclass{llncs}

\usepackage{amsmath, amssymb, enumitem, centernot, url, color,tikz}
\usepackage{graphicx,url,verbatim,bbm,enumitem,multirow,array,diagbox}
\usetikzlibrary{shapes,arrows,backgrounds,fit,positioning}
\usetikzlibrary{chains, decorations.pathreplacing, positioning}
\usetikzlibrary{calc}
\usetikzlibrary{arrows.meta}

\DeclareFontFamily{U}{matha}{}
\DeclareFontShape{U}{matha}{m}{n}{
  <-5.5>    matha5
  <5.5-6.5> matha6 
  <6.5-7.5> matha7
  <7.5-8.5> matha8
  <8.5-9.5> matha9
  <9.5-11>  matha10
  <11->     matha12
}{}
\DeclareSymbolFont{matha}{U}{matha}{m}{n}
\DeclareFontSubstitution{U}{matha}{m}{n}
\DeclareFontFamily{U}{mathx}{\hyphenchar\font45}
\DeclareFontShape{U}{mathx}{m}{n}{<-> mathx10}{}
\DeclareSymbolFont{mathx}{U}{mathx}{m}{n}
\DeclareFontSubstitution{U}{mathx}{m}{n}

\DeclareMathDelimiter{\ldbrack}{4}{matha}{"76}{mathx}{"30}
\DeclareMathDelimiter{\rdbrack}{5}{matha}{"77}{mathx}{"38}
\DeclareMathSymbol{\bigovoid}{1}{mathx}{"EC}

\numberwithin{equation}{section}
\spnewtheorem{myclaim}{Claim}{\bfseries}{}

\newcommand{\functions}{\mathrm{F}}

\newcommand{\dH}{d_\mathrm{H}}

\newcommand{\asynchronous}[1]{#1_{\mathrm{Asy}}}
\newcommand{\sequential}[1]{#1_{\mathrm{Seq}}}
\newcommand{\pparallel}[1]{#1_{\mathrm{Syn}}}

\newcommand{\semigroup}[1]{\left\langle #1 \right\rangle}
\newcommand{\asynchronousSemigroup}[1]{\asynchronous{\semigroup{#1}}}
\newcommand{\sequentialSemigroup}[1]{\sequential{\semigroup{#1}}}
\newcommand{\parallelSemigroup}[1]{\pparallel{\semigroup{#1}}}

\newcommand{\Sing}{\mathrm{Sing}}

\newcommand{\Sym}{\mathrm{Sym}}
\newcommand{\Alt}{\mathrm{Alt}}

\renewcommand{\(}{\ldbrack}
\renewcommand{\)}{\rdbrack}

\newcommand{\BF}[1]{{\bf\boldmath{#1}\unboldmath}}

\newcommand{\N}{\mathbb{N}}

\newcommand\coords[1]{[1,#1]}

\title{On simulation in automata networks\thanks{This work was partially funded by the CNRS and Royal Society joint research project PRC1861 and the French ANR project FANs ANR-18-CE40-0002 and the ECOS project C16E01.}}

\author{Florian Bridoux\inst{1} \and Maximilien Gadouleau\inst{2} \and Guillaume Theyssier\inst{3}}

\institute{
  Universit\'e Aix-Marseille, CNRS, LIS, Marseille, France
  \email{florian.bridoux@lis-lab.fr}
  \and
  Department of Computer Science, Durham University, Durham, UK
  \email{m.r.gadouleau@durham.ac.uk}
  \and
  Universit\'e Aix-Marseille, CNRS, I2M, Marseille, France
  \email{guillaume.theyssier@cnrs.fr}
}

\tikzset{>=latex}

\begin{document}

\maketitle

\begin{abstract}
  An automata network is a finite graph where each node holds a state from some finite alphabet and is equipped with an update function that changes its state according to the configuration of neighboring states. More concisely, it is given by a finite map $f:Q^n\rightarrow Q^n$. 
  In this paper we study how some (sets of) automata networks can be simulated by some other (set of) automata networks with prescribed update mode or interaction graph. 
  Our contributions are the following. For non-Boolean alphabets and for any network size, there are intrinsically non-sequential transformations (i.e. that can not be obtained as composition of sequential updates of some network). Moreover there is no universal automaton network that can produce all non-bijective functions via compositions of asynchronous updates. On the other hand, we show that there are universal automata networks for sequential updates if one is allowed to use a larger alphabet and then use either projection onto or restriction to the original alphabet. We also characterize the set of functions that are generated by non-bijective sequential updates. Following Tchuente, we characterize the interaction graphs $D$ whose semigroup of transformations is the full semigroup of transformations on $Q^n$, and we show that they are the same if we force either sequential updates only, or all asynchronous updates.  
\end{abstract}

\section{Introduction}

An automata network is a network of entities each equipped with a local update function that changes its state according to the states of neighboring entities.
Automata networks have been used to model different kind of networks: gene networks, neural networks, social networks, or network coding (see \cite{GR16} and references therein). They can also be considered as a model of distributed computation with various specialized definitions \cite{WR79i,WR79ii}. The architecture of an automata network can be represented via its interaction graph, which indicates which update functions depend on which variables. An important stream of research is to determine how the interaction graph affects different properties of the network or to design networks with a prescribed interaction graph and with a specific dynamical property (see \cite{Gad18b} for a review of known results). On the other hand, automata networks are usually associated with an update mode describing how local update functions of each entity are applied at each step of the evolution. In particular, three categories of update modes can be distinguished: sequential (one node update at a time), asynchronous (any subset of nodes at a time) or synchronous (all nodes simultaneously). Studying how changing the update mode affects the properties of an automata network with fixed local update functions is another major trend in this field \cite{NS17,GN12,GOLES2016118}. Comparing the computational power of sequential and parallel machines is of course at the heart of computer science, but the questioning on update modes is also meaningful for applications of automata networks in modeling of natural systems where the synchronous update mode is often considered unrealistic.

For both parameters (interaction graphs and update modes), the set of properties that could be potentially affected is unlimited. In this paper, instead of choosing a set of properties to analyze, we adopt an intrinsic approach: we study how some (sets of) automata networks can be simulated by some other (set of) automata networks with prescribed update mode or interaction graph.

\paragraph{Notations.} We will always consider alphabets of the form ${\(q\)=\{0,\ldots,q-1\}}$ for some $q$ and usually denote by $n$ the number of nodes of the network which are identified by integers in the interval ${\coords{n}}$.
An \textbf{automata network} is a map $f : \(q\)^n \to \(q\)^n$. An element ${x\in\(q\)^n}$ is a configuration and ${x_v}$ denotes the state of node $v$ in configuration $x$. By extension $f_v$ denotes the map ${x\mapsto f(x)_v}$. The rank of $f$ is the size of its image. For any set of coordinates ${V\subseteq\coords{n}}$, ${f^{(V)} : \(q\)^n \to \(q\)^n}$ denotes the following map: 
\[f^{(V)}(x)_i =
  \begin{cases}
    f(x)_i&\text{if }i\in V\\
    x_i&\text{else.}
  \end{cases}
\]
The notation is extended to words of subsets $w = (w_1, \dots, w_k)$ as follows: $f^{(w)} = f^{(w_k)} \circ \cdots \circ f^{(w_1)}$. For $v\in\coords{n}$ we overload this notation by ${f^{(v)}=f^{(\{v\})}}$. 

We will often consider semigroups of functions under compositions: $\semigroup{X}$ where $X$ is a set of functions that denotes the semigroup generated by compositions of elements of $X$. We denote the fact that $S_1$ is a sub-semigroup of $S_2$ by ${S_1\le S_2}$. For any set $X$, ${\Sym(X)}$ is the set of permutations on $X$.
We denote the set of all networks $f : \(q\)^n \to \(q\)^n$ as $\functions(n,q)$. We denote by ${\Sym(n,q)}$ the set of ${f\in\functions(n,q)}$ which are bijective and by ${\Sing(n,q)}$ the set of ${f\in\functions(n,q)}$ which are non-bijective. For any set $F$ of functions in $\functions(n, q)$, what they can simulate (asynchronously, sequentially, synchronously) is denoted as follows:
\begin{align*}
	\asynchronousSemigroup{F} &:= \semigroup{ \left\{ f^{(V)} : f \in F, V \subseteq \coords{n} \right\} },\\
	\sequentialSemigroup{F} &:= \semigroup{ \left\{ f^{(v)} : f \in F, v \in \coords{n} \right\} },\\
	\parallelSemigroup{F} &= \semigroup{F}.
\end{align*}
Then we say that $F$ \BF{simulates} $g \in \functions(n,q)$ asynchronously (sequentially, synchronously, respectively) if $g \in \asynchronousSemigroup{F}$ ($\sequentialSemigroup{F}$, $\parallelSemigroup{F}$, respectively). When ${F=\{f\}}$ we use notations $\asynchronousSemigroup{f}$, $\sequentialSemigroup{f}$, $\parallelSemigroup{f}$, respectively.

\paragraph{Previous works.} Simulation of automata networks is the topic of two main strands of work. The first stream investigates what a single network can simulate. The main observation, made in \cite{BCG19}, is that there is no sequentially complete network for $\functions(n,q)$, i.e. for all $f \in \functions(n,q)$, $\sequentialSemigroup{f} \ne \functions(n,q)$. This was refined in several ways. Firstly, there is no sequentially complete network for singular (\textit{i.e.} non-permutation) transformations: for all $f \in \functions(n,q)$, $\Sing(n,q) \not\le \sequentialSemigroup{f}$ \cite{BCG19}. Secondly, for all $n \ge 2$ and $q \ge 2$ (unless $n=q=2$), there exists a sequentially complete network for permutations: there exists $f \in \functions(n,q)$ such that $\sequentialSemigroup{f} = \Sym(n,q)$ \cite{CFG14}. These results illustrate a clear dichotomy between permutations and non-permutations. Thirdly, the simulation model was extended in \cite{BCG19} to include situations whereby a large network $f \in \functions(m,q)$ could simulate a smaller network $g \in \functions(n,q)$ for $n \le m$; notably, there always exists a complete network of size $m = n+1$ which can sequentially simulate any $g \in \functions(n,q)$. 

Another strand of work considers simulation by (possibly large) sets of networks. Firstly, Tchuente \cite{Tch86} investigated what networks with a prescribed reflexive interaction graph $D$ could simulate synchronously. The main result is that this set of networks $\functions(D,q)$ is complete, i.e. $\parallelSemigroup{ \functions(D,q) } = \functions(n,q)$, if and only if $D$ is strongly connected and has a vertex of in-degree $n$. Secondly, in the context of in-situ computation (a.k.a. memoryless computation), Burckel proved that any network could be sequentially simulated, if we allow the updates to differ at each time step; in our language: for all $n$ and $q$, $\sequentialSemigroup{ \functions(n, q) } = \functions(n,q)$ \cite{Bur96}. This seminal result was subsequently refined (see \cite{BGT14,GR15}); notably linear bounds on the shortest word required to simulate a transformation were obtained in \cite{BGT09,BGT14}.

\paragraph{Our contributions.} In this paper, we are further developing both strands of the theory of simulation of automata networks. We make the following contributions. We first consider simulation by a single network. Firstly, we show that for any $q \ge 3$ and any $n \ge 2$, there exists a network $g \in \functions(n,q)$ which is not sequentially simulatable. Secondly, we consider asynchronous simulation, and we show that there is no asynchronously complete network: for all $f \in \functions(n,q)$, $\Sing(n,q) \not\le \asynchronousSemigroup{f}$. This is a clear strengthening of the result in \cite{BCG19} for sequential simulation. Thirdly, we extend the framework to let a network over a large alphabet $f \in \functions(n,q')$ simulate a network $g \in \functions(n,q)$ over a smaller alphabet. We consider two ways to extend the alphabet, and for each we prove the existence of sequentially complete networks for $q' = 2q$ and $q' = q+1$, respectively. We then consider simulation by large sets of networks. The seminal result in \cite{Bur96} shows that instructions (updates of the form $f^{(v)}$ for some $v \in [1,n]$) can simulate any network; in this paper, we determine what singular instructions can simulate (and even idempotent instructions for $q \ge 3$). We finally strengthen the main result in \cite{Tch86} by showing that it also holds when considering sequential and asynchronous updates as well.

\paragraph{Proofs.} Complete proofs of all lemmas, propositions and theorems can be found in \cite{DBLP:journals/corr/abs-2001-09198}.

\section{Sequential simulation}
\label{sec:seq}

We say $g \in \functions(n,q)$ is \textbf{sequentially simulatable} if $g \in \sequentialSemigroup{f}$ for some $f \in \functions(n,q)$. Recall that unless $n=q=2$ any ${g\in\Sym(n,q)}$ is sequentially simulatable since there is a universal $f \in \functions(n,q)$ such that $\sequentialSemigroup{f} = \Sym(n,q)$ \cite{CFG14}. Concerning non-bijective maps, the situation is radically different for non-Boolean alphabets as shown in the following theorem.
\newcommand\orphans[1]{O(#1)}
For any function $\phi\in\functions(n,q)$, we denote by $\orphans(\phi)$ the set of its orphans: ${\orphans(\phi)=\{c \in \(q\)^n : \phi^{-1}(c)=\emptyset\}}$. The analysis of oprhans configurations under sequential updates is the key behind the following theorem.

\begin{theorem} \label{thm:nonboolean}
For any $n \ge 2$ and $q \ge 3$, there exists $h \in \functions(n,q)$ which is not sequentially simulatable.
\end{theorem}

The functions which are not sequentially simulatable produced in the proof of Theorem~\ref{thm:nonboolean} have two configurations $a$ and $b$ in ${\(q\)^n}$ with the same image and another $d$ which is an orphan with the following property: for each coordinate $i$ where $a_i$ and $b_i$ differ, $d_i$ is different from both $a_i$ and $b_i$. Note that this situation is impossible in the Boolean case since if ${a_i\neq b_i}$ then necessarily ${d_i\in\{a_i,b_i\}}$.

F. Bridoux did an exhaustive search in $\functions(n,2)$ with $n=2$ and $n=3$ to test which one are sequentially simulatable \cite{these_florian}. It turns out that all $f\in\functions(3,2)$ are sequentially simulatable.
However, some functions in $\functions(2,2)$ are not and one example is the circular permutation ${00\rightarrow 01\rightarrow 11\rightarrow 10\rightarrow 00}$ \cite[Proposition 12]{these_florian}.
More details (including the code of the test program) are available at \url{http://theyssier.org/san2020}.

\section{Asynchronous simulation}
\label{sec:async}

In this section, we consider asynchronous simulation, where at each step we allow any update $f^{(T)}$ for $T \subseteq [1,n]$. We then refine the result in \cite{BCG19} that there is no network that can sequentially simulate all singular networks. 

We say that a function $h : B \to C$, where $B$ and $C$ are finite sets, is balanced if for any $c, c' \in C$, $|h^{-1}(c)| = |h^{-1}(c')|$. In particular, if $f \in \functions(n,q)$ is bijective, then all its coordinate functions $f_v : \(q\)^n \to \(q\)$ must be balanced.

\begin{theorem}\label{thm:noasyncsing}
For all $f \in \functions(n,q)$, $\Sing(n,q) \not\le \asynchronousSemigroup{f}$.
\end{theorem}

\begin{proof}
Suppose, for the sake of contradiction, that $\Sing(n,q) \le \asynchronousSemigroup{f}$. We first show that not all coordinate functions of $f$ are balanced. There exists $S \subseteq [1,n]$ such that $f^{(S)}$ has rank $q^n - 1$. (Otherwise, no function in $\asynchronousSemigroup{f}$ has rank $q^n-1$.) Then there exist $a, b \in \(q\)^n$ such that
\[
	\left| \left( f^{(S)} \right)^{-1}(x) \right| = \begin{cases}
	2 & \text{if } x = a\\
	0 & \text{if } x = b\\
	1 & \text{otherwise}.
	\end{cases}
\]
Then let $v \in S$ such that $a_v \ne b_v$. We have 
\[
	|f_v^{-1}(a_v)| = \sum_{x : x_v = a_v} \left| \left( f^{(S)} \right)^{-1}(x) \right| = 2 + \sum_{x : x_v = a_v, x \ne a} 1 = q^{n-1} + 1,
\]
thus $f_v$ is not balanced.

Thus, suppose $f_v$ is not balanced, and let $q_0 \in \(q\)$ such that $|f_v^{-1}(q_0)| < q^{n - 1}$. Say a network $h \in \functions(n,q)$ is defective if $h^{-1}(x) = \emptyset$ for some $x$ with $x_v = q_0$. Let $g \in \Sing(n,q)$ not be deficient, and have a nontrivial $g_v$; and suppose $g = f^{(w_1 \cdots w_k)}$. Let $i = \max\{ 1 \le j \le k : v \in w_j \}$, then $f^{(w_i)}$ is defective, and so is $f^{(w_1 \cdots w_i)}$. Since $f^{(w_{i+1} \cdots w_k)}$ fixes the coordinate $v$, $f^{(w_1 \cdots w_k)} = g$ is also deficient, which is the desired contradiction.\qed
\end{proof}

Similarly to Theorem \ref{thm:nonboolean}, the obstacle in Theorem \ref{thm:noasyncsing} was found in the set of maps of rank $q^n-1$. We now show that maps of rank $q^n-2$ form another obstruction to having complete simulation in the asynchronous case. Let $T(n,q)$ be the set of networks in $\functions(n,q)$ whose rank is not equal to $q^n - 1$. It is clear that $T(n,q)$ is a semigroup, generated by maps of rank $q^n$ or $q^n - 2$.

\begin{proposition}
For all $f \in \functions(n,q)$, $T(n,q) \not\le \asynchronousSemigroup{f}$.
\end{proposition}

\begin{proof}
Suppose, for the sake of contradiction, that $T(n,q) \le \asynchronousSemigroup{f}$. Firstly, all the coordinate functions of $f$ are balanced. Indeed, let $g(x) = x + (1, \dots, 1)$ and express $g = f^{(w_1 \cdots w_k)}$. Then $f^{(w_i)}$ is bijective and hence $f_v$ is balanced for all $v \in w_i$; since $\bigcup_{i=1}^k w_i = [1,n]$, we obtain that $f_v$ is balanced for all $v \in [1,n]$. Secondly, the proof of Theorem \ref{thm:noasyncsing} showed that there is no $f^{(S)}$ of rank $q^n-1$.

Now, there are two types of networks with rank $q^n - 2$:
\begin{itemize} 
	\item 
	Say $g$ is of type I if there exists $a \in \(q\)^n$ such that $|g^{-1}(a)| = 3$ (and hence any other $x \ne a$ has $|g^{-1}(x)| \le 1$). 

	\item 
	Say $h$ is of type II if there exist $a, b \in \(q\)^n$ such that $|h^{-1}(a)| = |h^{-1}(b)| = 2$ (and hence any other $x \notin \{a,b\}$ has $|h^{-1}(x)| \le 1$).
\end{itemize}
By an argument similar to the proof of Theorem \ref{thm:noasyncsing}, there is no $S \subseteq [1,n]$ such that $f^{(S)}$ is of type I. Let $g$ be of type I and let us express it as $g = f^{(w_1 \cdots w_k)}$. Each $f^{(w_l)}$ has rank at least $q^n - 2$, and there exists $1 \le i \le k$ such that $f^{(w_i)}$ is singular. By the argument above, $f^{(w_i)}$ is of type II and so is $h := f^{(w_1 \cdots w_i)}$, say $|h^{-1}(a)| = |h^{-1}(b)| = 2$. Denote $g = h' \circ h$ for $h' := f^{(w_{i+1} \cdots w_k)}$. If $h'(a) = h'(b)$, then $g$ has rank at most $q^n - 3$; otherwise $|g^{-1}(h'(a))| = |g^{-1}(h'(b))| = 2$ and hence $g$ is of type II, which is the desired contradiction.\qed
\end{proof}

\section{Simulation using larger alphabets}
\label{sec:largeralphabet}

\newcommand{\1}{{\normalfont \texttt{1}}}
\newcommand{\0}{{\normalfont \texttt{0}}}

As said earlier, there is no universal automata network in $\functions(n,q)$ able to sequentially simulate all functions of $\functions(n,q)$ (actually Theorem~\ref{thm:noasyncsing} gives a stronger negative result). In this section, we revisit this problem when the simulator is allowed to use a larger alphabet. In this case we can consider two natural types of simulations: one requires the simulation to work on any initial configuration of the simulator and uses a projection onto configurations of the simulated functions; the other does not use projection, but works only on initial configurations using the alphabet of the simulated function.

\begin{definition}
  Let $n\in\N$, ${2\leq q<q'}$ and consider ${f\in\functions(n,q')}$. We say that $f$ is \emph{($n$,$q$)-universal by factor} if there is a surjection ${\pi:\(q'\)\rightarrow\(q\)}$ such that for any $h \in \functions(n,q)$ there is a word $w \in \coords{n}^\ast$ such that 
  \[\forall x\in\(q'\)^n, \overline\pi \circ f^{(w)}(x) = h\circ \overline\pi(x)\]
  where ${\overline\pi(x_1,\ldots,x_n)=(\pi(x_1),\ldots,\pi(x_n))}$.  $f$ is said \emph{($n$,$q$)-universal by initialization} if for any $h\in\functions(n,q)$ there is a word $w\in\coords{n}^\ast$ such that
   \[ \forall x\in \(q\)^ n,  f^{(w)} (x)= h(x).\]
\end{definition}

We are going to show that universality can be achieved for each kind of simulation. In both cases, the larger alphabet allows us to encode more information than the configuration of the simulated function. This additional information is used as a global controlling state that commands transformations applied on the simulated configuration and evolves according to a finite automaton. In the case of simulation by factor, the encoding is straightforward but the global controlling state is uninitialized. The key is to use a control automaton with a synchronizing word (see Figure~\ref{fig_projected_complete}). In the case of simulation by initialization, the difficulty lies in the encoding.

The following theorems were obtained by F. Bridoux during his PhD thesis \cite{these_florian}.

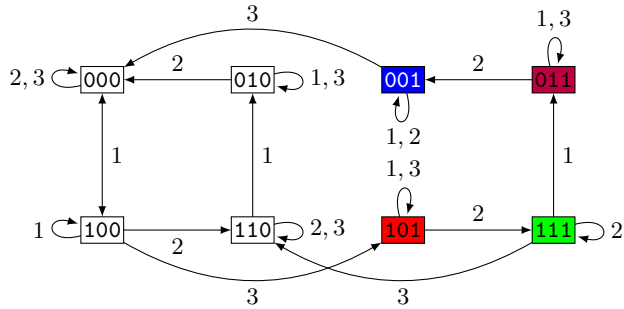
\begin{figure}
	\centering
	\begin{tikzpicture}[scale=1]
	\tikzstyle{grossefleche} = [-{>[length=1.5mm]}]
	\tikzstyle{config} = [draw,outer sep=0,inner sep=1,minimum size=10,fill=white]
	\tikzstyle{configOrange} = [draw,outer sep=0,inner sep=1,minimum size=10,fill=orange]
	\node[config]    (vooo) at (-3,2) {$\0\0\0$};
	\node[config]       (voio) at (-1,2) {$\0\1\0$};
	\node[config]       (vioo) at (-3,0) {$\1\0\0$};
	\node[config]       (viio) at (-1,0) {$\1\1\0$};
	\node[config,fill=blue,text=white]    	(vooi) at (1,2) {$\0\0\1$};
	\node[config,fill=purple]       (voii) at (3,2)  {$\0\1\1$};
	\node[config,fill=red]       (vioi) at (1,0)  {$\1\0\1$};
	\node[config,fill=green]       (viii) at (3,0)  {$\1\1\1$};
	
	\draw[grossefleche] (vooo) edge[loop left] node[left]{$2,3$}  (vooo);
	\draw[grossefleche] (voio) edge[loop right] node[right]{$1,3$}  (voio);	
	\draw[grossefleche] (vioo) edge[loop left] node[left]{$1$}  (vioo);	
	\draw[grossefleche] (viio) edge[loop right] node[right]{$2,3$}  (viio);	
	
	\draw[grossefleche] (vooi) edge[loop below] node[below]{$1,2$}  (vooi);
	\draw[grossefleche] (voii) edge[loop above] node[above]{$1,3$}  (voii);
	\draw[grossefleche] (vioi) edge[loop above] node[above]{$1,3$}  (vioo);		
	\draw[grossefleche] (viii) edge[loop right] node[right]{$2$}  (viii);
		
	\draw[grossefleche] (voio) edge[] node[above]{$2$}  (vooo);
	\draw[grossefleche,{<[length=1.5mm]}-{>[length=1.5mm]}] (vioo) edge[] node[right]{$1$}  (vooo);
	\draw[grossefleche] (vioo) edge[] node[below]{$2$}  (viio);
	\draw[grossefleche] (viio) edge[] node[right]{$1$}  (voio);
	\draw[grossefleche] (vioi) edge[] node[above]{$2$}  (viii);
	\draw[grossefleche] (viii) edge[] node[right]{$1$}  (voii);
	\draw[grossefleche] (voii) edge[] node[above]{$2$}  (vooi);
	
	\draw[grossefleche] (vooi) edge[bend right] node[above]{$3$}  (vooo);
	\draw[grossefleche] (viii) edge[bend left] node[below]{$3$}  (viio);
	\draw[grossefleche] (vioo) edge[bend right] node[below]{$3$}  (vioi);
	\end{tikzpicture}
	
	\caption{Definition and sequential behavior of $\rho: \(2\)^3 \to \(2\)^3$ from Theorem~\ref{theorem:complet_par_facteur}. Label on arcs represent the coordinate updated.}
	\label{fig_projected_complete}
\end{figure}

\begin{theorem} \label{theorem:complet_par_facteur}
  For any $q \geq 2$ and $n \geq 3$, there exists $f \in \functions(n,2q)$ which is $(n,q)$-universal by factor.
\end{theorem}
\begin{proof}
  We can see any configuration of ${\(2q\)^n}$ as a pair made of a configuration of $\(q\)^n$ and a Boolean configuration, so we can as well describe $f$ as a function acting on ${\(q\)^n\times \(2\)^n}$ to simplify notations and use the surjective map ${\pi : \(q\)^n\times\(2\)^n\rightarrow\(q\)^n}$ that projects onto the first component. We will actually choose $f$ which is the identity map on the coordinates $4$ to $n$ on the Boolean component. So, to simplify even further, we will define a function ${f:\(q\)^n\times\(2\)^3\rightarrow\(q\)^n\times\(2\)^3}$.
  
  Consider first the function $\rho:\(2\)^3 \to \(2\)^3$ defined by Figure~\ref{fig_projected_complete} and consider the map $\Psi: \(q\)^n\times\(2\)^3 \to \(q\)^n$ defined by:
  \begin{align*}
    \Psi(x,y)_1 &= \begin{cases}
      x_1 + \1 \bmod q & \text{if } y = \color{red} \1\0\1 \color{black},\\
      \1 & \text{if }  x= (\0)^n \text{ and }  (y =  \color{purple}  \0\1\1 \color{black} \text{ or } y = \color{blue}  \0\0\1 \color{black}),\\
      \0 & \text{if }  x= \1(\0)^{n-1} \text{ and }  y = \color{purple}  \0\1\1 \color{black},\\
      x_1 & \text{otherwise},
    \end{cases}\\
    \Psi(x,y)_2 &= \begin{cases}
      x_2 + \1 \bmod q & \text{if } x_1 = \0 \text{ and }  y =  \color{green} \1\1\1 \color{black},\\
      x_2 & \text{otherwise},
    \end{cases}\\
    \Psi(x,y)_3 &= \begin{cases}
      x_3 + \1 \bmod q & \text{if } x_1 = x_2 = \0 \text{ and }  y =  \color{purple}  \0\1\1 \color{black},\\
      x_3 & \text{otherwise},
    \end{cases}\\
    \forall i \in [4,n],\ \Psi(x,y)_i &= \begin{cases}
      x_i + \1 \bmod q & \text{if } x_1 = x_2 = \dots = x_{i-1} = \0,\\
      x_i & \text{otherwise}.
    \end{cases}
  \end{align*}
	
  Then we define $f$ by ${f(x,y) = \bigl(\Psi(x,y),\rho(y)\bigr)}$. We
  now prove properties about $f$ implying that it is $(n,q)$-universal
  by factor.

  \begin{myclaim}\label{claim:forcesecondcomponent}
    For any ${(x,y)\in\(q\)^n\times\(2\)^3}$ it holds $f^{((3)^q,2,3,1,1,2,1,3)}(x,y) = (x,\color{red}\1\0\1\color{black})$.
  \end{myclaim}
  \begin{proof}
    First, let us remark that updating $q$ times coordinate 3 starting from ${(x,y)}$, there are two cases:
    \begin{itemize}
    \item $y \neq \color{purple} \0\1\1 \color{black}$ or  $x_1 \neq \0$ or $x_2 \neq \0$ and then the component $x$ is not modified;
    \item $y = \color{purple} \0\1\1 \color{black}$ and $x_1 = x_2 = \0$, and then the modification $x_3 \gets x_3 +1$ is applied $q$ times.
    \end{itemize}
    Therefore we have ${f^{((3)^q)}(x,y) = (x',y')}$ with \[x'=(x_1,x_2,x_3+q,\dots) = (x_1,x_2,x_3,\dots) = x.\]
    To show that the update sequence $((3)^q,2,3,1,1,2,1,3)$ does not modify the component $x$, it is sufficient to verify the following:
    \begin{itemize}
    \item coordinate 1 is not updated when $y \in \{ \color{red} \1\0\1 \color{black} , \color{purple} \0\1\1 \color{black}, \color{blue} \0\0\1 \color{black} \}$;
    \item coordinate 2 is not updated when $y =  \color{green} \1\1\1 \color{black}$;
    \item when coordinate 3 is updated and $y = \color{purple} \0\1\1 \color{black}$, it is updated $q$ times.
    \end{itemize}
    
    By definition of $f^{((3)^q,2,3,1,1,2,1,3)}$, we obtain:
    \[
      \begin{array}{lllllllllll}
        x\0\0\0 & \xrightarrow{(3)^q} x\0\0\0 &\xrightarrow{2} x\0\0\0 &\xrightarrow{3} x\0\0\0 &\xrightarrow{1} x\1\0\0 &\xrightarrow{1} x\0\0\0
        &\xrightarrow{2} x\0\0\0 &\xrightarrow{1} x\1\0\0 &\xrightarrow{3} x\color{red} \1\0\1 \color{black}, \\
        x\1\0\0 & \xrightarrow{(3)^q} x\color{red} \1\0\1 \color{black} &\xrightarrow{2} x\color{green} \1\1\1 \color{black} &\xrightarrow{3} x\1\1\0 &\xrightarrow{1} x\0\1\0 &\xrightarrow{1} x\0\1\0 &\xrightarrow{2} x\0\0\0 &\xrightarrow{1} x\1\0\0 &\xrightarrow{3} x\color{red} \1\0\1 \color{black},\\
        x\0\1\0 & \xrightarrow{(3)^q} x\0\1\0 &\xrightarrow{2} x\0\0\0 &\xrightarrow{3} x\0\0\0 &\xrightarrow{1} x\1\0\0 &\xrightarrow{1} x\0\0\0
        &\xrightarrow{2} x\0\0\0 &\xrightarrow{1} x\1\0\0 &\xrightarrow{3} x\color{red} \1\0\1 \color{black},\\
        x\1\1\0 & \xrightarrow{(3)^q} x\1\1\0 &\xrightarrow{2} x\1\1\0 &\xrightarrow{3} x\1\1\0 &\xrightarrow{1} x\0\1\0 &\xrightarrow{1} x\0\1\0
        &\xrightarrow{2} x\0\0\0 &\xrightarrow{1} x\1\0\0 &\xrightarrow{3} x\color{red} \1\0\1 \color{black},\\
        x\color{blue} \0\0\1 \color{black}  & \xrightarrow{(3)^q} x\0\0\0 &\xrightarrow{2} x\0\0\0 &\xrightarrow{3} x\0\0\0 &\xrightarrow{1} x\1\0\0 &\xrightarrow{1} x\0\0\0
        &\xrightarrow{2} x\0\0\0 &\xrightarrow{1} x\1\0\0 &\xrightarrow{3} x\color{red} \1\0\1 \color{black},\\
        x\color{purple} \0\1\1 \color{black}  & \xrightarrow{(3)^q} x\color{purple} \0\1\1 \color{black} &\xrightarrow{2} x\color{blue} \0\0\1 \color{black}  &\xrightarrow{3} x\0\0\0 &\xrightarrow{1} x\1\0\0 &\xrightarrow{1} x\0\0\0
        &\xrightarrow{2} x\0\0\0 &\xrightarrow{1} x\1\0\0 &\xrightarrow{3} x\color{red} \1\0\1 \color{black},\\
        x\color{red} \1\0\1 \color{black} & \xrightarrow{(3)^q} x\color{red} \1\0\1 \color{black} &\xrightarrow{2} x\color{green} \1\1\1 \color{black} &\xrightarrow{3} x\1\1\0 &\xrightarrow{1} x\0\1\0 &\xrightarrow{1} x\0\1\0
        &\xrightarrow{2} x\0\0\0 &\xrightarrow{1} x\1\0\0 &\xrightarrow{3} x\color{red} \1\0\1 \color{black},\\
        x\color{green} \1\1\1 \color{black} & \xrightarrow{(3)^q} x\1\1\0 &\xrightarrow{2} x\1\1\0 &\xrightarrow{3} x\1\1\0 &\xrightarrow{1} x\0\1\0 &\xrightarrow{1} x\0\1\0
        &\xrightarrow{2} x\0\0\0 &\xrightarrow{1} x\1\0\0 &\xrightarrow{3} x\color{red} \1\0\1 \color{black}.\\
	
      \end{array}
    \]
    \qed
  \end{proof}
  
  Let us now show that, starting from ${(x,\color{red}\1\0\1\color{black})}$, $f$ can realize three kinds of transformations on $x$ that will turn out to be sufficient to generate all ${\functions(n,q)}$.
  \begin{itemize}
  \item Let $c\in\Sym(n,q)$ be the following circular permutation: 
    \[
      c: ((\0)^n \to \1 (\0)^{n-1} \to \dots \to (q-1) (\0)^{n-1}  \to \0 \1 (\0)^{n-2} \to \dots).
    \]
    then for any ${x\in \(q\)^n}$ we have $f^{(1,2,2,1,(3,4,\dots,n))}(x,\color{red}\1\0\1\color{black}) = (c(x),\color{purple} \0\1\1 \color{black})$ because:
    \[
      \begin{array}{ll}
        x\color{red}\1\0\1\color{black} 
        \xrightarrow{1} c^{(1)}(x) \color{red}\1\0\1\color{black} 
        \xrightarrow{2} c^{(1)}(x) \color{green} \1\1\1 \color{black} 
        \xrightarrow{2} c^{([1,2])}(x) \color{green} \1\1\1 \color{black} 
        \xrightarrow{1} c^{([1,2])}(x) \color{purple} \0\1\1 \color{black} \\
        \xrightarrow{3} c^{([1,3])}(x) \color{purple} \0\1\1 \color{black} 
        \xrightarrow{4} c^{([1,4])}(x) \color{purple} \0\1\1 \color{black} 
        \xrightarrow{5} \dots 
        \xrightarrow{n} c(x) \color{purple} \0\1\1 \color{black}.
      \end{array}
    \]
  \item Consider the transposition $k = ((\0)^n \leftrightarrow \1(\0)^{n-1})$, then we have, for any ${x\in \(q\)^n}$, $f^{(2,1,1)}(x,\color{red}\1\0\1\color{black}) = (k(x),\color{purple} \0\1\1 \color{black})$ because:
    \[
      x \color{red}\1\0\1\color{black} 
      \xrightarrow{2} x \color{green} \1\1\1 \color{black}  
      \xrightarrow{1} x \color{purple} \0\1\1 \color{black}  
      \xrightarrow{1} k(x) \color{purple} \0\1\1 \color{black}.
    \]
  \item Finally, consider the assignment $d = ((\0)^n \to \1(\0)^{n-1})$, then for any ${x\in \(q\)^n}$ it holds $f^{(2,1,2,1)}(x,\color{red}\1\0\1\color{black}) = (d(x),\color{purple} \color{blue}  \0\0\1 \color{black})$ because:
    \[
      \begin{array}{ll}
        x \color{red}\1\0\1\color{black} 
        \xrightarrow{2} x \color{green} \1\1\1 \color{black} 
        \xrightarrow{1} x \color{purple} \0\1\1 \color{black} 
        \xrightarrow{2} x \color{blue}  \0\0\1 \color{black} 
        \xrightarrow{1} d(x) \color{blue}  \0\0\1 \color{black}.
      \end{array}
    \]
  \end{itemize}
  Since functions $c$, $k$ and $d$ generate $\functions(n,q)$ (see \cite{How95} or \cite{GM09}), the theorem follows.\qed
\end{proof}

\begin{theorem}  \label{theorem:complet_avec_initialisation}
  For any  $q \geq 2$ and $n \geq 3q$, there is $f \in \functions(n,q+1)$ which is $(n,q)$-universal by initialization.
\end{theorem}

\section{Simulation by sets of networks}

So far we studied what a single function can simulate. We know shift our interest to semigroups generated by some sets of functions.

\subsection{Singular instructions}

An instruction is any $f^{(v)}$ for some $f \in \functions(n,q)$ and some $v \in [1,n]$. Burckel showed that any network is the composition of instructions: $\sequentialSemigroup{\functions(n,q)} = \semigroup{ \left\{  f^{(v)} : f \in \functions(n,q), v \in [1,n] \right\} } = \functions(n,q)$. As an immediate consequence, any permutation in $\Sym(n,q)$ is the composition of permutation instructions: $\Sym(n,q)$ is exactly $\semigroup{ \left\{ f^{(v)} \in \Sym(n,q) : f \in \functions(n,q), v \in [1,n] \right\} }$. We now determine what singular instructions generate: let
\[
	S(n,q) := \semigroup{ \left\{ f^{(v)} \in \Sing(n,q) : f \in \functions(n,q), v \in \coords{n} \right\} }.
\]

\begin{proposition} \label{prop:S(q,n)}
The semigroup $S(n,q)$ generated by singular instructions consists of all networks $f$ such that there exist $a,b \in \(q\)^n$ with $f(a) = f(b)$ and $\dH(a,b) = 1$.
\end{proposition}


Any network $f$ can be seen as a vertex colouring of the Hamming graph $H(n,q)$ ($x$ colored by $f(x)$). From the proposition above, networks in $S(n,q)$ correspond to improper colouring. Since the chromatic number of $H(n,q)$ is equal to $q$, we deduce that any network with rank at most $q-1$ can be generated by singular instructions. However, the network
	${f(x) = \left(x_1 + \ldots + x_n, 0, \ldots, 0 \right)}$
cannot be generated by singular instructions, since it generates a proper colouring of the Hamming graph.

A network $f$ is idempotent if $f^2 = f$. Idempotents are pivotal in the theory of semigroups, for they are the identity elements of the subgroups of a given semigroup. In particular it is interesting to know whether a semigroup $S$ is generated by its set of idempotents, because then any element $s \in S$ can be expressed as a product of consecutively distinct idempotents: $s = e_1 e_2 \dots e_k$.
We remark that if $f \in S(n,q)$ is idempotent and has rank $q^n - 1$, then it must be an assignment instruction.

\begin{theorem} \label{th:S(q,n)}
$S(n,q)$ is generated by assignment instructions for $q \ge 3$.
\end{theorem}

The previous result could be proved using the so-called fifteen-puzzle. In the original puzzle, an image is cut into a four-by-four grid of tiles; one of the tiles is removed, thus creating a hole; the remaining fifteen tiles are scrambled by sliding a tile into  the hole. The player is then given the scrambled image, and has to reconstruct it by repeatedly sliding a tile in the hole. 

Clearly, this game can be played on any simple graph $D$, where a hole is created at a vertex (say $h$), and one can ``slide'' one vertex into the hole, the hole thus moving to that vertex. 
If the hole goes back to its original place $h$, then we have created a permutation of $V(D) \setminus h$. The set of all possible permutations is closed under composition and hence it forms a group, called the \BF{puzzle group}\index{group!puzzle} $G(D,h)$.  Wilson \cite{Wil74} fully characterised that group for $2$-connected simple graphs; we give a simpler version of the theorem below. 

\begin{theorem}[Wilson's fifteen-puzzle theorem]
Let $D$ be a $2$-connected simple graph, then $G(D,h) \cong G(D,h')$ for all vertices $h, h' \in V(D)$. Moreover, if $D$ is the undirected cycle, then $G(D,h)$ is trivial. Otherwise, the following hold.
\begin{enumerate}
	\item If $D$ is not bipartite and has at least eight vertices, then $G(D,h) = \Sym(V(D) \setminus h)$.
	
	\item If $D$ is bipartite, then $G(D,h) = \Alt(V(D) \setminus h)$.
\end{enumerate}
\end{theorem}

Using assignment instructions $(a \to b)$ to simulate a network $f$ of rank $q^n - 1$ can be viewed as playing the fifteen-puzzle on the Hamming graph $H(n,q)$: the first $(a^1 \to b^1)$ places a hole in vertex $a^1$ and any subsequent $(a^k \to b^k)$ slides the vertex $a^k$ into the hole $b^k$ (and the hole moves to $a^k$ instead). Since $H(n,q)$ is not bipartite for $q \ge 3$ (and it has at least nine vertices for $n \ge 2$), we can apply Wilson's theorem and, after a bit more work, prove Theorem \ref{th:S(q,n)} that way. However, the hypercube $H(n,2)$ is bipartite, then the puzzle group is only the alternating group. Thus, $S(n,2)$ is not generated by assignment instructions, and in particular $f = (010\cdots0 \leftrightarrow 110\cdots0) \circ (000\cdots0 \to 100\cdots0)$ cannot be generated by assignment instructions.

\subsection{Simulation by graphs}

The \textbf{interaction graph} of $f \in \functions(n,q)$ is the (directed graph) which has vertex set $V = \coords{n}$ and has an arc from $u$ to $v$ if and only if $f_v$ depends essentially on $u$, i.e. there exists $a,b \in \(q\)^n$ such that $a_{V\setminus u} = b_{V\setminus u}$ and $f_v(a) \ne f_v(b)$. For any graph $D$ with $n$ nodes, we denote the set of networks in $\functions(n,q)$ whose interaction graph is a subgraph of $D$ as $\functions(D,q)$.

A graph is reflexive if for any vertex $v$, $(v,v)$ is an arc in $D$. Note that for any reflexive graph $D$ it holds
	${\sequentialSemigroup{ \functions(D, q) } \subseteq \asynchronousSemigroup{ \functions(D, q) } = \parallelSemigroup{ \functions(D, q) }.}$
The first inclusion is trivial; the equality follows from the fact that for any $f \in \functions(D,q)$ and any $S \subseteq [1,n]$, $f^{(S)}$ belongs to $\functions(D,q)$ as well. 
Moreover, it is clear that if $\sequentialSemigroup{ \functions(H, q) } = \functions(n,q)$, then $H$ is reflexive (otherwise, $\sequentialSemigroup{\functions(H,q)}$ would not contain any permutation). The reflexive graphs which can simulate the whole of $\functions(n,q)$ synchronously were classified by Tchuente in \cite{Tch86}. In fact, the same graphs can simulate the whole of $\functions(n,q)$ asynchronously or sequentially.

\begin{theorem} \label{th:graph_universal}
Let $D$ be a reflexive graph on $n$ vertices. Then the following are equivalent.
\begin{enumerate}
	\item \label{it:sequential_universal_graph} $\sequentialSemigroup{ \functions(D, q) } = \functions(n, q)$.
	
	\item \label{it:asynchronous_universal_graph} $\asynchronousSemigroup{ \functions( D, q) } = \functions(n, q)$.
	
	\item \label{it:parallel_universal_graph} $\parallelSemigroup{ \functions(D, q) } = \functions(n, q)$.
	
	\item \label{it:universal_graph} $D$ is strongly connected and it has a vertex of in-degree $n$.
\end{enumerate}
\end{theorem}

A permutation of variables is any network $f := \bar{\phi}$ defined by $f_i(x) = x_{\phi(i)}$ for some $\phi \in \Sym([1,n])$. We first show that we can permute variables freely if the graph is strongly connected (and is reflexive for the sequential case).

\begin{lemma} \label{lem:strong_all_transpositions}
The following are equivalent for a reflexive graph $D$.
\begin{enumerate}
	\item \label{it:sequential_all_transpositions} $\sequentialSemigroup{ \functions( D, q ) }$ contains all permutations of variables of $\functions(n,q)$.
	
	\item \label{it:asynchronous_all_transpositions} $\asynchronousSemigroup{ \functions( D, q ) }$ contains all permutations of variables of $\functions(n,q)$.
	
	\item \label{it:strong_all_transpositions} $D$ is strong.
\end{enumerate}
\end{lemma}

\begin{proof}[Proof of Theorem \ref{th:graph_universal}]
Clearly, \ref{it:sequential_universal_graph} implies \ref{it:asynchronous_universal_graph}, which in turn is equivalent to \ref{it:parallel_universal_graph}. We prove \ref{it:asynchronous_universal_graph} implies \ref{it:universal_graph}. Let $D$ such that $\asynchronousSemigroup{ \functions(D, q) } = \functions(n, q)$. By Lemma \ref{lem:strong_all_transpositions}, $D$ is strong. We now prove that $D$ has a vertex of in-degree $n$. Otherwise, let $f \in \functions(D,q)$ of rank $q^n - 1$. Let $a \in \orphans(f)$ and $b$ with $|f^{-1}(b)| = 2$ (and hence $|f^{-1}(x)| = 1$ for any other $x$). We then have
\[
	\sum_{x \in \(q\)^n} f(x) \bmod q^n= b - a \ne 0.
\]
On the other hand, it is easily seen that for any $y \in \(q\)$, $|f_v^{-1}(y)|$ is a multiple of $q^{n - d_v}$ where $d_v$ is the in-degree of $v$ in $D$, hence
\[
	\sum_{x \in \(q\)^n} f_v(x) \bmod q = \sum_{y \in \(q\)} |f_v^{-1}(y)| y \bmod q = 0.
\]
Doing this componentwise for all $v$, we obtain $\sum_{x \in \(q\)^n} f(x) = 0$,
which is the desired contradiction.

We prove \ref{it:universal_graph} implies \ref{it:sequential_universal_graph}. We only need to show that all instructions in $\functions(n,q)$ belong to $\sequentialSemigroup{ \functions(D, q) }$. Let $u$ be a vertex of in-degree $n$, then we already have any instruction updating $u$. Let $v$ be another vertex, and $g$ be an instruction updating $v$,
then $g = \overline{(u \leftrightarrow v)} \circ h \circ \overline{(u \leftrightarrow v)}$, where $h$ is the instruction updating $u$ such that $h_u = g_v \circ \overline{(u \leftrightarrow v)}$.
Then $\overline{(u \leftrightarrow v)} \in \sequentialSemigroup{ \functions(D, q) }$ according to Lemma \ref{lem:strong_all_transpositions}. Thus, any instruction can be generated.\qed
\end{proof}

\section{Future work}

The contrast between the complete sequential simulator for $\Sym(n,q)$ and the existence of non-bijective functions that are not sequentially simulatable in the non-Boolean case is striking. We would like first to settle the Boolean case: we conjecture that all functions of $\functions(n,2)$ are sequentially simulatable for large enough $n$. For ${q\geq 3}$, in order to better understand the set of sequentially simulatable networks, one could for instance analyze how much synchronism is required to simulate them (how large are the sets $V$ in the asynchronous updates $f^{(V)}$ used to simulate them). In particular, one may ask whether, for all $n$, there exists some network with $n$ entities that require a synchronous update $f^{([1,n])}$ in order to be simulated asynchronously. Besides, the networks considered in sections~\ref{sec:seq}, \ref{sec:async} and \ref{sec:largeralphabet} have an unconstrained interaction graph. The situation could be very different when restricting all networks to particular a family of interaction graphs (bounded degree, bounded tree-width, etc).  
Finally, still concerning interaction graphs, the characterization of Theorem~\ref{th:graph_universal} is about reflexive graphs. We would like to extend it to any graph (not necessarily reflexive).

\bibliographystyle{splncs04}
\bibliography{FDS}

\newpage
\appendix

\section{Proof of Theorem~\ref{thm:nonboolean}}

\begin{proof}
  Let $a,b,c,d$ such that $a_{[1,2]} = (0,0),\  b_{[1,2]} = (0,1),\ c_{[1,2]} = (1,0),\ d_{[1,2]} = (1,2)$ and $ a_{[3,n]} = b_{[3,n]} = c_{[3,n]} = d_{[3,n]} = (0)^{n-2}$.
  Let $h \in \functions(n,q)$ such that $h(b) = a,\ h(c) = b,\  h(d) = c$ and for all $x \in \(q\)^n\ \setminus\ \{b,c,d\}$, $h(x) = x$.
  So $d$ is the only orphan configuration and $h$ is of rank $q^n-1$: $\orphans(h)=\{d\}$.
  
  For the sake of contradiction, let us suppose that there exists $f \in \functions(n,q)$ and $w = w_1\cdots w_p \in \coords{n}^*$ such that $f^{(w)} = h$.
  \begin{myclaim}\label{claim:twoupdated}
    Coordinate $2$ is updated at least once in $w$ and $f^{(2)}$ is of rank ${q^n-1}$. Moreover, denoting by $k$ the minimum step in $w$ such that $w_k = 2$, then $f^{(w_j)}$ is bijective for any $j<k$. In particular $f^{(w_1, \dots, w_{k-1})}$ is bijective.
  \end{myclaim}
  \begin{proof}
    Since $h$ does not act like the identity on coordinate $2$ then coordinate $2$ must be updated at least once. Then the rank of $h$ is not larger than the rank of $f^{(2)}$. So the rank of $f^{(2)}$ is at least $q^n-1$. Consider now the smallest $k$ such that ${f^{(w_1\cdots w_k)}(a)=f^{(w_1\cdots w_k)}(b)}$ (it must exist since $h(a)=h(b)$). First, ${f^{(w_1\cdots w_{k-1})}(a)\neq f^{(w_1\cdots w_{k-1})}(b)}$ (it would contradict minimality of $k$) and ${f^{(w_1\cdots w_{k-1})}(\alpha)\neq f^{(w_1\cdots w_{k-1})}(\beta)}$ for any other pair $\alpha\neq\beta$ because otherwise it would create a second pair of configurations with the same image under $h$. Therefore $f^{(w_1\cdots w_{k-1})}$ is bijective and necessarily ${f^{(w_k)}}$ is not bijective. Moreover coordinate $w_k$ does not appear in ${w_1\cdots w_{k-1}}$ (it would contradict bijectivity of $f^{(w_1\cdots w_{k-1})}$) so ${f^{(w_1\cdots w_{k-1})}(a)_{w_k}=a_{w_k}}$ and ${f^{(w_1\cdots w_{k-1})}(b)_{w_k}=b_{w_k}}$. Since $a$ and $b$ differ only on coordinate $2$ we must have $w_k=2$ to ensure ${f^{(w_1\cdots w_k)}(a)=f^{(w_1\cdots w_k)}(b)}$ while ${f^{(w_1\cdots w_{k-1})}(a)\neq f^{(w_1\cdots w_{k-1})}(b)}$. The claim follows.\qed
  \end{proof}
  \begin{myclaim}\label{claim:uniqueorphan}
    For any ${\ell\geq k}$ there exists $y' \in \(q\)^n$ such that $\orphans(f^{(w_1\cdots w_\ell)}) = \{ y' \}$ and $y'_2 = d_2 = 2$.
  \end{myclaim}
  \begin{proof}
    From Claim~\ref{claim:twoupdated} we know that ${f^{(2)}}$ is of rank $q^n-1$, and in fact any ${f^{(w_1\cdots w_\ell)}}$ is of this rank. Let us denote by $\omega(\ell)$ the unique element of ${\orphans(f^{(w_1\cdots w_\ell)})}$ for each $\ell\geq k$. We will show by induction on $\ell$ that ${\omega(\ell)_2=\omega(k)_2}$. This allows to conclude the Claim since ${\orphans(h)=\orphans(f^{(w)})=\{d\}}$. The initialization of the induction is trivial. For the induction step, suppose it holds ${\omega(\ell)_2=\omega(k)_2}$ for some $|w|>\ell\geq k$. First, if $w_{\ell+1}=2$ then ${\orphans(f^{(w_1\cdots w_{\ell+1})}) = \orphans(f^{(2)})}$ because ${f^{(w_1\cdots w_{\ell+1})}=f^{(2)}\circ f^{(w_1\cdots w_\ell)}}$ and is of same rank as ${f^{(2)}}$. In this case we deduce immediately ${\omega(\ell+1)_2=\omega(k)_2}$. Suppose now that $j=w_{\ell+1}\neq 2$. Let $E := \{ e \in\(q\)^n : \omega(\ell)_{\coords{n} \setminus \{j\} } = e_{\coords{n} \setminus \{j\} }  \}$. Note that ${f^{(j)}(x)\in E}$ implies $x\in E$. There are two cases:
    
    \begin{enumerate}
    \item if $f^j$ restricted to the set
      $E$ is not a bijection, there exists
      $e \in E$ such that for any $b \in E, f^{(j)}(b) \neq e$. Then
      we have $e \in \orphans(f^{(j)})$. 
    \item if $f^{(j)}$
      restricted to the set $E$ is a bijection, then for all
      $x \in E$ there is a unique $y \in E$ such that
      $f^{(j)}(y) = x$. Let us take $e := f^{(j)}(\omega(\ell))$ which is in $E$ by definition. Then since
      $\omega(\ell) \in \orphans( f^{(w_1, \dots, w_{k+\ell})} )$, we have $e \in\orphans(f^{(w_1, \dots, w_{k+\ell+1})})$. 
    \end{enumerate}
    In both cases we have,
    $e \in \orphans(f^{(w_1, \dots, w_{k+\ell+1})})\cap E$.
    Thus, $\omega(\ell+1) = e \in E$ and as a result
    $\omega(\ell+1)_2 = e_2 = \omega(\ell)_2 $ which concludes the induction.\qed
  \end{proof}
  Let $y'$ be the only orphan of $f^{(w_1 \cdots w_k)}$, $z' := f^{(2)}(y')$ and let $z,y \in \(q\)^n$ be such that $f^{(w_1\cdots w_{k-1})}(z) = z'$ and $f^{(w_1\cdots w_{k-1})}(y) = y'$.
  Now, there are two cases:
  
  \begin{enumerate}
  \item Suppose there is $x' \in \(q\)^n \setminus \{ y'\}$, such that
    $f^{(2)}(x') = z'$ and let $x$ such that
    $f^{(w_1, \dots, w_{k-1})}(x) = x'$. As a result,
    ${f^{(w_1, \dots, w_k)}(y) = f^{(2)}(y') = z' = f^{(2)}(x') =
      f^{(w_1, \dots, w_k)}(x)}$. Thus,
    $h(y) = f^{(w)}(y) = f^{(w)}(x) = h(x)$. We then have $y \in \{a,b\}$, which is a contradiction since $y_2 = y'_2 = 2 \notin \{a_2, b_2\}$.
    
  \item 	Suppose $\bigl(f^{(2)}\bigr)^{-1}(z')=\{y'\}$.
    Let us first show that $2$ is updated only once in $w$.
    We actually show that if it is updated at least $2$ times then $h$ is of rank at most $q^n-2$. 
    So let $k'>k$ be such that $w_{k'}=2$. Consider the orbit of $y$ under $f^{(2)}$. Denote ${e(m) = \bigl(f^{(2)}\bigr)^{m-1}(y)}$ for ${1\leq m}$ and let $\ell$ be the last step before the orbit is cycling, \textit{i.e.} the smallest integer such that ${e(\ell+1)\in\{e(1),\ldots,e(\ell)\}}$.
    Say ${e(\ell+1)=e(i)}$.	We know that ${i>2}$ since $e(1) = y$ and $e(2) =z$.
    Furthermore, we have ${f^{(2)}(e(i-1)) = e(i) = f^{2}(e(\ell))}$ with ${e(i-1) \neq e(\ell)}$.
    We know that we have for any $j,j' \in [\ell]$, $e(j) \neq e(j') $ but they can only differ on coordinate $2$ by definition so $e^{(j)}_2 \neq e^{(j')}_2$.
    Let $y'' \in \(q\)^n$ be such that $\orphans(f^{(w_1\cdots w_{k'-1})}) = \{ y''\}$ and $y''_2 = y'_2$ as ensured by Claim~\ref{claim:uniqueorphan}. Since $e(1) = y'$ we have $e(\ell)_2 \neq y''_2$ and $ e(i-1)_2 \neq y''_2$ (configurations $e(m)$ differ on coordinate $2$ by definition).
    So $e(\ell)$ and $e(i-1)$ have pre-images $\alpha$ and $\beta$ under $f^{(w_1\cdots w_{k'-1})}$, but since ${f^{(2)}}$ sends both ${e(\ell)}$ and ${e(i-1)}$ to ${e(i)}$, then ${f^{(w_1\cdots w_{k'})}}$ sends both $\alpha$ and $\beta$ to ${e(i)}$ which means that the rank of ${f^{(w_1\cdots w_{k'})}}$ is at least one less than the rank of ${f^{(w_1\cdots w_{k'-1})}}$, that is at most ${q^n-2}$. Then $h$ would have rank at most $q^n-2$ which is a contradiction. We thus show that $2$ is only updated once in $w$, at step $k$. In particular we must have for any $x\in \(q\)^n$ ${h(x)_2 = f^{(w_1\cdots w_k)}(x)_2}$.
    If $a', b', c',d'$ are the respective images of $a,b,c,d$ by $f^{(w_1\cdots w_{k-1})}$, then the definition of $h$ forces ${f^{(2)}(x)=x}$ for any ${x\in \(q\)^n\setminus\{b',c',d'\}}$. From Claim~\ref{claim:twoupdated}, we know that $f^{(2)}$ is of rank $q^n-1$ and we must have ${f^{(2)}(a')=f^{(2)}(b')=a'}$ (otherwise we would create another pair of configurations than $a,b$ with the same image under $h$). This shows that $a'_{\coords{n} \setminus \{2\}} = b'_{\coords{n} \setminus \{2\}}$.	Furthermore, we prove that $f^{(2)}(c') = b'$.
    First, ${f^{(2)}(c')\not\in\{c',d'\}}$ because $f_2(c') = h_2(c) = 1$, $c'_2 = 0$ and $d'_2 = 2$.
    Second, if $f^{(2)}(c') \neq b'$, then $f^{(2)}(c')$ would belong to $\(q\)^n\ \setminus\ \{b',c',d'\}$ and, consequently, $f^{(w_1\cdots w_k)}(c) = f^{(2)}(c') = f^{(2)}(x') = f^{(w_1\cdots w_k)}(x)$ for some $x \in \(q\)^n$ and $h(x) = h(c)$. This is impossible as it would create another pair of configurations than $a,b$ with the same image under $h$.
    So it holds $f^{(2)}(c') = b'$ and $c'_{\coords{n} \setminus \{2\}} = b'_{\coords{n} \setminus \{2\}} = a_{\coords{n} \setminus \{2\}}$.
    However, $c'_2 = c_2 = 0 = a_2 = a'_2$.	We deduce $c' = a'$ and $h(c) = h(a)$, which contradicts the definition of $h$.
  \end{enumerate}
  The theorem follows.\qed
\end{proof}

\section{Proof of Theorem~\ref{theorem:complet_avec_initialisation}}

\begin{figure}
	\centering
	\begin{tikzpicture}[scale=1]
	\tikzstyle{grossefleche} = [-{>[length=1.5mm]}]
	\tikzstyle{config} = [draw,outer sep=0,inner sep=1,minimum size=10,fill=white]
	\tikzstyle{configOrange} = [draw,outer sep=0,inner sep=1,minimum size=10,fill=orange]

	\node[config,fill=orange]    (vooo) at (-3,2) {$\0\0\0$};
	\node[config,fill=blue,text=white]       (voio) at (-1,2) {$\0\1\0$};
	\node[config,fill=red]       (vioo) at (-3,0) {$\1\0\0$};
	\node[config]       (viio) at (-1,0) {$\1\1\0$};
	\node[config]    	(vooi) at (1,2) {$\0\0\1$};
	\node[config,fill=green]       (voii) at (3,2)  {$\0\1\1$};
	\node[config,fill=purple]       (vioi) at (1,0)  {$\1\0\1$};
	\node[config]       (viii) at (3,0)  {$\1\1\1$};
	
	\draw[grossefleche] (vooo) edge[] node[above]{$2+$}  (voio);
	\draw[grossefleche] (voio) edge[bend left] node[above]{$3+$}  (voii);
	\draw[grossefleche] (voii) edge[] node[below]{$2^-$}  (vooi);
	\draw[grossefleche] (vooi) edge[] node[right]{$1^+$}  (vioi);
	\draw[grossefleche] (vioi) edge[bend right] node[above]{$3^-$}  (vioo);
	\draw[grossefleche] (vioo) edge[] node[right]{$1^-$}  (vooo);
	\end{tikzpicture}
	
	\caption{Cycle of 3-bits configurations controlling the behavior of the automata network $f$ from Theorem~\ref{theorem:complet_avec_initialisation}.}
	\label{fg_initialised_complete}
\end{figure}
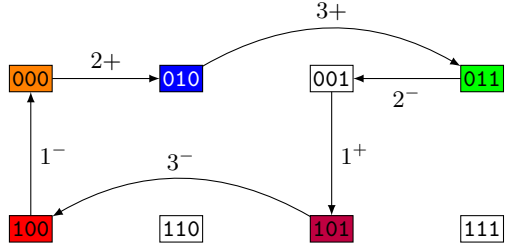
\begin{proof}
  The starting idea is to view a configuration $\(q+1\)^q$ as a configuration of $\(q\)^q$ together with a bit of information. To formalize this, let $Z^+$ be the set of configuration of ${\(q+1\)^q}$ with a single coordinate in state $q+1$, and let ${Z=\(q\)^q\cup Z^+}$. The Boolean value of some configuration is given by the map ${\lambda:\(q+1\)^q\rightarrow\(2\)}$ such that
$\lambda(z) = \begin{cases}
\0 &\text{ if } z \in \(q\)^q\\
\1 &\text{ otherwise}
\end{cases}.$ Moreover let's define an interpretation map ${\mu : \(q+1\)^q\to\(q\)^q}$  as follows:
\begin{itemize}
\item if $z \not\in Z$ then fix $\mu(z)$ arbitrarily (it doesn't matter below);
\item if $z \in \(q\)^q$ then $\mu(z) = z$;
\item finally, if $z \in Z^+$ then there is a unique  $j \in \coords{q}$ such that $z_j = q+1$ and we let ${(\mu(z))_{\coords{q} \setminus \{ j \}} = z_{\coords{q} \setminus \{ j \}}}$ and ${(\mu(z))_j = (j-1) - \sum\limits_{i \in \coords{q} \setminus \{ j \} } z_i}$.
\end{itemize}
Note that $\mu(Z^+) = \(q\)^q$. Indeed if ${x=(x_1,\ldots,x_q)}$ we let ${j-1=\sum_{i\in\coords{q}} x_i\bmod q}$ and
\[z=(x_1,\ldots,x_{j-1},q+1,x_{j+1},\ldots,x_q).\] Then one can check that ${\mu(z)=x}$. We now define two functions $\boxplus: Z \mapsto Z$ and $\boxminus: Z \mapsto Z$ which sends any element of $Z$ to $Z^+$ or $\(q\)^n$ respectively while preserving the image by $\mu$.  We define $\forall z \in Z$
\begin{align*}
\forall j \in \coords{q},\ \boxplus_j(z) &= \begin{cases}
q+1 &\text{if } z \in \(q\)^n \text{ and } j-1 =  \sum\limits_{ i \in \coords{q}  } z_{i}\bmod q,  \\
z_j &\text{otherwise}, 
\end{cases}\\
\forall j \in \coords{q},\ \boxminus_j(z) &= \begin{cases}
(j-1) - \sum\limits_{ i \in \coords{q} \setminus \{j\}  } z_i\bmod q &\text{if } z_i = q+1\\
z_j &\text{otherwise}.
\end{cases}
\end{align*}

First remark that for any ${z\in Z}$ it holds that $z'=\boxplus^{(1,2,\dots,q)}(z)\in Z^+$ verifies $\mu(z) = \mu(z')$.
Similarly,  $z''=\boxminus^{(1,2,\dots,q)}(z)\in\(q\)^n$ verifies $\mu(z) = \mu(z'')$.

Thanks to the above we can encode $3$ bits and a configuration of $\(q\)^n$ in any configuration of ${\(q+1\)^n}$ as follows (recall that ${n\geq 3q}$). Each bit is encoded using a bloc of nodes of size $q$ and the coding through $\lambda$ and $Z$ described above. The configuration of $\(q\)^n$ is read through $\mu$ in these three blocs and directly in the remaining nodes. More precisely, for any $i \in \coords{3}$ and $j \in \coords{q}$, define $\Phi^i_j = (i-1) q + j$.
For any $i \in \coords{3},$ define $\Phi^i = \{ \Phi^i_1, \dots \Phi^i_q \}$.
Finally, let $R = [3q+1,n]$, so that we have $\coords{n} = \Phi^1 \cup \Phi^2 \cup \Phi^3 \cup R$.
Consider now the function $\varphi: \(q+1\)^n \to \(q\)^n$ such that if ${z_R\in \(q\)^R}$ then ${\varphi(z)_{\Phi_p}=\mu(z_{\Phi_p})}$ for ${p=1,2,3}$ and ${\varphi(z)_R=z_R}$, and its value is defined arbitrarily for any other $z\in \(q+1\)^n$.
Moreover, define $\phi: \(q+1\)^n \to \(2\)^3$ as:
\[ 
\phi: z \mapsto  \lambda(z_{\Phi^1})  \lambda(z_{\Phi^2}) \lambda(z_{\Phi^3}). 
\]

Thus we interpret any ${z\in \(q+1\)^n}$ as an element of ${\(q\)^n\times\(2\)^3}$ by ${z\mapsto(\varphi(z),\phi(z))}$. Through this interpretation, our automata network will intuitively work as follows: the configuration component $\(q\)^n$ is the one on which we actually apply transformations, and the three bits serve as a global control state which governs the transformation we are applying at each step. The key property of our coding is  that any $z\in\(q\)^n$ is such that ${\phi(z)=000}$ and $\varphi(z)=z$. So a simulation by initialization actually means that our control state is initialized to $000$. Our construction is such that we only reach configurations $z'\in\(q+1\)^n$ such that $z'_R\in\(q\)^R$.

We first define the map $\Psi: \(q+1\)^n \to \(q\)^n$ which describes how the configuration component is modified depending on the control states. It behaves roughly as a $q$-ary counter, except for some coordinates in the blocs of nodes encoding the three bits.
For any $z\in\(q+1\)^n$, and for any $i\in\coords{n},$ we denote $x=\varphi(z)$ and $y=\phi(z)$ and let
\begin{align*}
\text{if } i = \Phi^1_1,\ \Psi_i(z) &= \begin{cases}
x_i + 1\bmod q & \text{if } y= \color{blue} \0\1\0 \color{black}, \\
1 & \text{if } y= \color{green} \0\1\1 \color{black} \text{ and } x=(\0)^n , \\
0 & \text{if } y= \color{green} \0\1\1 \color{black} \text{ and } x=1(\0)^{n-1}, \\
x_i & \text{else},
\end{cases} \\
\text{else if } i = \Phi^2_1,\ \Psi_i(z) &= \begin{cases}
x_i + 1\bmod q  & \text{if } y = \color{red}  \0\0\1 \color{black} \text{ and } x_1 = x_2 = \dots = x_{i-1} = \0,  \\
1 & \text{if }  y = \color{purple} \1\0\1 \color{black} \text{ and } x = (\0)^n,  \\
x_i & \text{else},
\end{cases} \\
\text{else if } i = \Phi^3_1,\ \Psi_i(z) &= \begin{cases}
x_i + 1\bmod q  & \text{if } y = \color{orange} \0\0\0 \color{black} \text{ and }  x_1 = x_2 = \dots = x_{i-1} = \0, \\
x_i & \text{else},
\end{cases} \\
\text{else, } \Psi_i(z) &= \begin{cases}
x_i + 1\bmod q & \text{if } x_1 = x_2 = \dots = x_{i-1} = 0,\\
x_i & \text{else.}
\end{cases} \\
\end{align*}

We now can finally define the automata network $f \in \functions(n,q')$ as follows: apply $\Psi$ on the configuration component and update the three control bits according to Figure~\ref{fg_initialised_complete}. Precisely:
\begin{itemize}
\item for each arc $ y \xrightarrow{i^+} y'$ in Figure~\ref{fg_initialised_complete} and each $z \in \(q+1\)^n$ such that 
$\phi(z) = y$ 
,  $f_{\Phi^i}(z) =  \boxplus(z_{\Phi^i})$ if $z_{\Phi^i} \in Z$ (and an arbitrary value otherwise);
\item for each arc $ y \xrightarrow{i^-} y'$ in Figure~\ref{fg_initialised_complete} and each $z \in \(q+1\)^n$ such that  
$\phi(z) = y$
, $f_{\Phi^i}(z) =  \boxminus(z_{\Phi^i})$ if $z_{\Phi^i} \in Z$ (and an arbitrary value otherwise);
\item for any other case $f_j(z) =  \Psi_j(z)$.
\end{itemize}

For any $i \in \coords{3}$, let's define the update sequence $\tilde{i} = \Phi^i_1\Phi^i_2\cdots \Phi^i_q$ and let $\tilde{4} = (3q+1)\cdots n$.

First, it holds that for any $x \in \(q\)^n$, the function $f^{(\tilde{2}\tilde{1}\tilde{3}\tilde{2}\tilde{1}\tilde{3}\tilde{2}\tilde{1}\tilde{3}\tilde{4})}$ maps $x\color{orange}\0\0\0\color{black}$ to $c(x)\color{orange}\0\0\0\color{black}$ with $c$ the circular permutation from the proof of Theorem~\ref{theorem:complet_par_facteur}. Indeed:
\[
\begin{array}{ll}
x\color{orange} \0\0\0 \color{black}
\xrightarrow{\tilde{2}} x \color{blue} \0\1\0 \color{black} 
\xrightarrow{\tilde{1}} c^{(\coords{q})}(x) \color{blue} \0\1\0 \color{black}
\xrightarrow{\tilde{3}} c^{(\coords{q})}(x) \color{green} \0\1\1 \color{black}
\xrightarrow{\tilde{2}} c^{(\coords{q})}(x) \color{black} \0\0\1 \color{black}
\xrightarrow{\tilde{1}} c^{(\coords{q})}(x) \color{purple} \1\0\1 \color{black}\\
\xrightarrow{\tilde{3}} c^{(\coords{q})}(x) \color{red} \1\0\0 \color{black}
\xrightarrow{\tilde{2}} c^{(\coords{2q})}(x) \color{red} \1\0\0 \color{black}
\xrightarrow{\tilde{1}} c^{(\coords{2q})}(x) \color{orange} \0\0\0 \color{black}
\xrightarrow{\tilde{3}} c^{(\coords{3q})}(x) \color{orange} \0\0\0 \color{black}
\xrightarrow{\tilde{4}} c(x) \color{orange} \0\0\0 \color{black}.
\end{array}
\]

Next, for any $x \in \(q\)^n$, the function $f^{(\tilde{2}, \tilde{3}, \tilde{1}, \tilde{2}, \tilde{1}, \tilde{3}, \tilde{1})}$ maps $x\color{orange}\0\0\0\color{black}$ to $k(x)\color{orange}\0\0\0\color{black}$ where $k$ is the transposition $ ((\0)^n \leftrightarrow \1 (\0)^{n-1})$. Indeed:
\[
\begin{array}{ll}
x \color{orange} \0\0\0 \color{black} 
\xrightarrow{\tilde{2}} x\color{blue} \0\1\0 \color{black}
\xrightarrow{\tilde{3}} x\color{green} \0\1\1 \color{black}
\xrightarrow{\tilde{1}} k(x)\color{green} \0\1\1 \color{black}
\xrightarrow{\tilde{2}} k(x) \color{black} \0\0\1 \color{black}
\xrightarrow{\tilde{1}} k(x) \color{purple} \1\0\1 \color{black}
\xrightarrow{\tilde{3}} k(x) \color{red} \1\0\0 \color{black}
\xrightarrow{\tilde{1}} k(x) \color{orange} \0\0\0 \color{black}.
\end{array}
\]

Finally, for any $x \in \(q\)^n$, the function $f^{(\tilde{2}, \tilde{3}, \tilde{2}, \tilde{1}, \tilde{2}, \tilde{3}, \tilde{1})}$ maps $x\color{orange}\0\0\0\color{black}$ to $d(x)\color{orange}\0\0\0\color{black}$ with $d$ the assignment $ ((\0)^n \to (\0)^q \1 (\0)^{n-q-1})$. Indeed:
\[
\begin{array}{ll}
(x\color{orange} \0\0\0 \color{black}) 
\xrightarrow{\tilde{2}} x\color{blue} \0\1\0 \color{black}
\xrightarrow{\tilde{3}} x\color{green} \0\1\1 \color{black}
\xrightarrow{\tilde{2}} x\color{black} \0\0\1 \color{black}
\xrightarrow{\tilde{1}} x\color{purple} \1\0\1 \color{black}
\xrightarrow{\tilde{2}} x\color{purple} \1\0\1 \color{black}
\xrightarrow{\tilde{3}} x\color{red} \1\0\0 \color{black}
\xrightarrow{\tilde{1}} d(x) \color{orange} \0\0\0 \color{black}.
\end{array}
\]	

The theorem follows because function $c$, $k$ and $d$ generate $\functions(n,q)$.\qed
\end{proof}

\section{Proof of Proposition~\ref{prop:S(q,n)}}

\begin{proof}
Let $f \in S(n,q)$ and express it as $f = h \circ g$ where $g$ is a singular instruction, that updates the coordinate $v$. Then there exist $a,b$ such that $g(a) = g(b)$; they must satisfy $a_k = b_k$ for all $k \ne v$. Thus $f(a) = f(b)$, where $a$ and $b$ only differ in one coordinate.

Conversely, let $f$ satisfy $f(a) = f(b)$, where $a$ and $b$ only differ in one coordinate, say $v$. Then there exists $h$ such that $f = h \circ (a \to b)$ and instructions $g^2, \dots, g^L$ such that $f = g^L \circ \cdots \circ g^2 \circ (a \to b)$.

\begin{myclaim}
For any permutation instruction $p$ and any singular transformation $g$, there exists a singular instruction $s$ updating the same coordinate as $p$ such that $s \circ g = p \circ g$.
\end{myclaim}

\begin{proof}
Say $p$ updates coordinate $u$. Let $z$ be an orphan of $g$ and define the instruction $s$ updating $u$ as 
\[
	s_u(x) = \begin{cases}
	p_u(z) + 1	&\text{if } x = z,\\
	p_u(x) 		&\text{otherwise.}
	\end{cases}
\]
Since $p$ is a permutation, there exists $y \in \(q\)^n$ which only differs from $z$ on coordinate $u$ such that $p_u(y) = s_u(z)$. Then $s$ is singular (since $s(y) = s(z)$) and $s(g(x)) = p(g(x))$ for any $x \in \(q\)^n$.\qed
\end{proof}

Now convert any permutation instruction $g^i$ ($2 \le i \le L$) into the corresponding singular instruction $h^i$ such that 
\[
	g^i \circ (g^{i-1} \circ \dots \circ g^2 \circ (a \to b)) = h^i \circ (g^{i-1} \circ \dots \circ g^2 \circ (a \to b)).
\]
We thus express $f$ as a composition of singular instructions only.\qed
\end{proof}

\section{Proof of Theorem~\ref{th:S(q,n)}}

\begin{proof}
Let $f$ be an instruction in $\Sing(n,q)$; without loss of generality, say it updates the first coordinate. We first prove that it can be generated by instructions of rank $q^n-1$. For any $z \in \(q\)^n$, we denote $z_{-1} = (z_2, \dots, z_n)$ and $\hat{z} = \{y \in \(q\)^n: y_{-1} = z_{-1}\}$. Let $a,b \in \(q\)^n$ satisfy $f(a) = f(b)$; we remark that $a_{-1} = b_{-1}$ and hence $b \in \hat{a}$. We then have $f = g \circ (a \to b)$, where 
\[
	g(x) = \begin{cases}
	a &\mbox{if } x = a\\
	f(x) &\mbox{otherwise}
	\end{cases}
\]
is an instruction that fixes $a$.

For every $y \in \(q\)^{n-1}$, we denote the instruction $g^y$ such that
\[
	g^y(x) = \begin{cases}
	g(x) & \mbox{if } x_{-1} = y\\
	x & \mbox{otherwise;}
	\end{cases}.
\]
Clearly, all $g^y$ commute with each other, hence we can write $g = \bigovoid_{y \in \(q\)^{n-1}} g^y$ and
\[
	f = \left( \bigovoid_{y \in \(q\)^{n-1}} g^y \right) \circ (a \to b).
\] 
Moreover, if $y \ne a_{-1}$, then $g^y$ commutes with $(a \to b)$; since the latter is idempotent, we can distribute it as follows:
\[
	f = (g^{a_{-1}} \circ (a \to b)) \circ  \bigovoid_{y \ne a_{-1}} \left( g^y \circ (a \to b) \right).
\]
Each $g^y$ can be viewed as a transformation of $\(q\)$ acting on $\hat{y}$; as such, it can be expressed as a composition of transpositions and assignments. If $y = a_{-1}$, then $g^{a_{-1}}$ is singular, hence it can be expressed as a composition of assignments. Thus, the first term $g^{a_{-1}} \circ (a \to b)$ is the composition of assignment instructions. Any other term $g^y \circ (a \to b)$ is the composition of assignment instructions and of instructions of the form $[a,b,w,x] := (a \to b) \circ (w \leftrightarrow x) = (w \leftrightarrow x) (a \to b)$, where $w_{-1} = x_{-1} \ne a_{-1}$.

All that is left to prove is that $[a,b,w,x]$ is a composition of assignment instructions. Let $z$ such that $z_{-1} = w_{-1} = x_{-1}$ and $z_1 \notin \{ w_1, x_1 \}$ (such a $z$ exists since $q \ge 3$). Moreover, let $z^0 = a$, $z^1 = (z_1, a_2, \dots, a_n)$, and so on until $z^n = z$. We remark that $z^i$ and $z^{i+1}$ only differ in at most coordinate, and that $z^i$ is never equal to $w$ or $x$. Then we can compute $[a,b,w,x]$ as follows:
\[
	[a,b,w,x] = (a \to z^1) \circ \cdots \circ (z \to x) \circ (x \to w) \circ (w \to z) \circ (z \to z^{n-1}) \circ \cdots \circ (z^1 \to a) \circ (a \to b).
\]\qed
\end{proof}

\section{Proof of Lemma~\ref{lem:strong_all_transpositions}}

\begin{proof}
Clearly, \ref{it:sequential_all_transpositions} implies \ref{it:asynchronous_all_transpositions}. We prove \ref{it:asynchronous_all_transpositions} implies \ref{it:strong_all_transpositions}. If $D$ is not strong, then suppose there is no path from $u$ to $v$ in $D$. It is easy to see that for any $g \in \asynchronousSemigroup{\functions(D,q)}$, $g_v$ does not depend on $x_u$. Therefore, $\overline{(u \leftrightarrow v)} \notin \asynchronousSemigroup{\functions(D,q)}$.

We prove \ref{it:strong_all_transpositions} implies \ref{it:sequential_all_transpositions}. We prove that $\overline{(u \leftrightarrow v)} \in \sequentialSemigroup{ \functions( D, q ) }$ for any $u$, $v$ such that $(u,v) \in E(D)$. There is a path from $v$ back to $u$, say $v, w_1,w_2,\ldots,w_l,u$. The instructions simulating $\overline{(u \leftrightarrow v)}$ can be represented as follows. Any instruction $g$ updating coordinate $v$ is written $y_v \gets g_v(y)$; we keep track of what $g_v(y)$ actually means in terms of the original configuration $x$ on the right hand side. Composition is done top to bottom:
\begin{alignat*}{2}
	y_v 	&\gets y_u + y_v 				&\qquad & x_u + x_v\\
	y_{w_2} &\gets y_{w_2} + y_{w_1} 		&\qquad & x_{w_2} + x_{w_1}\\
	y_{w_3} &\gets y_{w_3} + y_{w_2} 		&\qquad & x_{w_3} + x_{w_2} + x_{w_1}\\
	&\vdots &&\\
	y_u 	&\gets y_u + y_{w_l} 			&\qquad & x_u + x_{w_l} + \cdots + x_{w_1}\\
	y_{w_l} &\gets y_{w_l} - y_{w_{l-1}} 	&\qquad & x_{w_l}\\
	&\vdots &&\\
	y_{w_2} &\gets y_{w_2} - y_{w_1}		&\qquad & x_{w_2}\\
	y_{w_1} &\gets y_{w_1} + y_v			&\qquad & x_{w_1} + x_u + x_v\\
	y_{w_2} &\gets y_{w_2} + y_{w_1} 		&\qquad & x_{w_2} + x_{w_1} + x_u + x_v\\
	&\vdots &&\\
	y_{w_l}	&\gets y_{w_l} + y_{w_{l-1}} 	&\qquad & x_{w_l} + \cdots + x_{w_1} + x_u + x_v\\
	y_u		&\gets y_{w_l} - y_u			&\qquad & x_v\\
	y_{w_l} &\gets y_{w_l} - y_{w_{l-1}} 	&\qquad & x_{w_l}\\
	&\vdots &&\\
	y_{w_1} &\gets y_{w_1} - y_v			&\qquad & x_{w_1}\\
	y_v 	&\gets y_v - y_u				&\qquad & x_u
\end{alignat*}
Second, since $D$ is strong, its undirected version has a spanning tree, hence we can generate all permutations of variables.\qed
\end{proof}

\end{document}